\documentclass[letterpaper, 10 pt, conference]{cssconf}

\IEEEoverridecommandlockouts
\usepackage{amsmath}
\usepackage{amssymb}
\usepackage{hyperref}
\usepackage{booktabs}
\usepackage{graphicx}
\usepackage{tikz}
\usetikzlibrary{calc,shapes.geometric}

\newtheorem{remark}{Remark}
\newtheorem{problem}{Problem}
\newtheorem{definition}{Definition}
\newtheorem{ansatz}{Ansatz}
\newtheorem{proposition}{Proposition}
\newtheorem{example}{Example}
\newtheorem{lemma}{Lemma}
\newtheorem{condition}{Condition}

\newcommand{\deviation}[0]{\epsilon}

\newcommand{\Natural}[0]{\mathbb{N}}
\newcommand{\Whole}[0]{\mathbb{Z}}
\newcommand{\Real}[0]{\mathbb{R}}

\newcommand{\domT}[0]{\Natural}
\newcommand{\domX}[0]{\mathbb{X}}
\newcommand{\domU}[0]{\mathbb{U}}
\newcommand{\domY}[0]{\mathbb{Y}}
\newcommand{\data}[0]{\mathcal{D}}
\newcommand{\eval}[0]{\delta}

\newcommand{\outputspace}[0]{\mathcal{Y}}
\newcommand{\hypothesis}[0]{\mathcal{H}}
\newcommand{\prehypothesis}[0]{\mathcal{G}}
\newcommand{\sequencespace}[0]{\mathcal{S}}

\newcommand{\encoder}[0]{\mathsf{E}}
\newcommand{\encodertarget}[0]{\mathsf{E}_\target}
\newcommand{\model}[0]{\mathsf{M}}
\newcommand{\target}[0]{\mathsf{T}}
\newcommand{\world}[0]{\mathsf{L}}
\newcommand{\identity}[0]{\mathsf{I}}

\newcommand{\modelspace}[0]{\mathcal{M}}
\newcommand{\targetspace}[0]{\mathcal{T}}
\newcommand{\regularization}[0]{\mathsf{R}}
\newcommand{\error}[0]{\mathrm{e}}
\newcommand{\HS}[0]{\mathrm{HS}}
\NewDocumentCommand\Lp{O{2}O{}}{\ensuremath{L^{#1}_{#2}}}

\makeatletter
\let\orig@begintheorem\@begintheorem
\let\orig@opargbegintheorem\@opargbegintheorem
\def\@begintheorem#1#2{\par\addvspace{3pt plus 1pt minus 1pt}\orig@begintheorem{#1}{#2}}
\def\@opargbegintheorem#1#2#3{\par\addvspace{3pt plus 1pt minus 1pt}\orig@opargbegintheorem{#1}{#2}{#3}}
\def\@endtheorem{\endtrivlist\unskip\addvspace{3pt plus 1pt minus 1pt}}
\makeatother

\title{\LARGE \bf
Demystifying Linear Operator Learning for Control Systems}

\author{Max Beier, Nicolas Hoischen, Sandra Hirche, Petar Bevanda
\thanks{M. Beier, N. Hoischen, S. Hirche, and P. Bevanda are with the
Chair of Information-Oriented Control,
School of Computation, Information and Technology,
Technical University of Munich, 80333 Munich, Germany
        {\tt\small \{max.beier, nicolas.hoischen, hirche, petar.bevanda\}@tum.de}}%
}

\begin{document}

\maketitle
\thispagestyle{empty}
\pagestyle{empty}

\begin{abstract}
This paper proposes a structured approach to learning linear operators for control systems from data.
We address both structural and learning-theoretic aspects of the problem.
To derive structural assumptions, we propose using the well-established framework of (semi)groups for evolution equations, as operators in control systems are of the same type.
Further, we propose analyzing learning algorithms through the lens of the inverse problems framework.
This reveals how a learned model depends on the data via error decompositions, convergence guarantees, and optimal regularization -- enabling us to compare existing methods and derive \emph{provably} advantageous algorithms.
In order to obtain these results, we restrict our scope to bounded operators on Hilbert spaces.
Although this may appear restrictive, existing approaches often make this assumption implicitly to obtain matrix-like representations.
We demonstrate the power of using these frameworks by deriving a convergent estimator for time-varying systems.
\end{abstract}
\newlength\notewd
\newcommand\notetext{This is the authors unedited version of a paper accepted to the 65th IEEE Conference on Decision and Control.

©2026 IEEE. Personal use of this material is permitted. Permission from IEEE must be obtained for all other uses, in any current or future media, including reprinting/republishing this material for advertising or promotional purposes, creating new collective works, for resale or redistribution to servers or lists, or reuse of any copyrighted component of this work in other works.}
\settowidth\notewd{\notetext}
\ifdim\notewd>\dimexpr\textwidth-2pt\relax
\setlength\notewd{\dimexpr\textwidth-2pt\relax}
\fi

\begin{tikzpicture}[remember picture, overlay]
\node[anchor=north, align=flush left, inner xsep=2pt, outer sep=0pt,
      text width=\notewd, font=\footnotesize]
    at ([xshift=\dimexpr 1in+\hoffset+\oddsidemargin+0.5\textwidth\relax,
         yshift=-\dimexpr 1in+\voffset+\topmargin+\headheight+\headsep+\textheight+0.3cm\relax]
        current page.north west)
    {\notetext};
\end{tikzpicture}
\section{INTRODUCTION}
Making informed decisions requires an understanding of the effects one's actions will have.
This can be achieved by modeling the system one is trying to influence.
When a precise first-principles model is unavailable, one often resorts to data to uncover dependencies~\cite{bruntonDataDrivenScienceEngineering2019}.
In that case, one collects a dataset comprising observations, external inputs, and the system's response to those inputs.
This dataset can be used to build a model of the data-generating process, which, in turn, can be used to predict the effects of actions and thus make decisions.

Model classes, their properties, and decision-making from data have been extensively studied in systems and control theory~\cite{astromFeedbackSystemsIntroduction2008}.
The data described above is usually assumed to be generated by dynamics in a state space which contains the system's memory, such that a static function of the state can predict future states.
Direct access to a state is not required, allowing observables to be defined and evaluated on any domain, as long as it captures the dependence of output at different time steps.
In terms of partially observed deterministic systems, this can be achieved via delay embeddings~\cite{takensDetectingStrangeAttractors1981,bruntonChaosIntermittentlyForced2017}.
From a statistical viewpoint, this property is called sufficient statistics~\cite{kallenbergFoundationsModernProbability2021}.
Extensive learning and control theory exist for state-space models.

This paper will depart from this paradigm and investigate a different view of data and decisions. Embedding data in a feature space is at the core of machine learning~\cite{ingosteinwartSupportVectorMachines2008,bishopDeepLearningFoundations2024,bachLearningTheoryFirst2024} and is immensely successful across data modalities. Motivated by this success, we model the data as generated by evaluating an embedding and its evolved version. To allow for arbitrarily high-dimensional embeddings, we employ theory for infinite-dimensional feature spaces.
This view allows for defining a pure evolution model in feature space as a linear operator, which can be studied using techniques from functional analysis and group theory~\cite{aubinAppliedFunctionalAnalysis2000, engelOneParameterSemigroupsLinear2000}.

Viewing control systems from an operator-theoretic perspective holds promise for constructing novel automatic control algorithms with guarantees for systems that were previously only amenable to highly specific nonlinear or simplifying linear analysis, see~\cite{bevandaKoopmanOperatorDynamical2021} and references therein.
When the data is generated in a state space, the evolution operator has the special structure of a composition operator and is called the Koopman operator~\cite{koopmanHamiltonianSystemsTransformation1931}.
These models have gained widespread attention in the control community due to their (bi)linear structure, which is advantageous for optimization-based control~\cite{strasserOverviewKoopmanbasedControl2026}.
Motivated by the interest of control theorists in practicable models and the immense success of feature spaces in machine learning, we want to revisit the first principles of building operator models from data.
These first principles build on the theory of evolution operators~\cite{engelOneParameterSemigroupsLinear2000} and inverse problems~\cite{englRegularizationInverseProblems1996,stuartInverseProblemsBayesian2010}.
This approach contrasts with most of the control literature, where the operator model is defined in terms of a state-space model to `linearize' it~\cite {bevandaKoopmanOperatorDynamical2021, strasserOverviewKoopmanbasedControl2026}, and data and its distribution are treated as afterthoughts.

Within the general framework we outlined, we will demonstrate that an evolution operator admitting a composition structure via a state space is a special case of the evolution operator approach.
We recover existing models from the literature and pose control-theoretic questions in terms of the operator models.
When learning models from data, we show that commonly used estimators are ill-posed or suboptimal, and we propose estimators that admit provably optimal convergence to the best-approximate model.
In summary, our goal is to provide a framework that allows one to describe modeling approaches in a unified way and obtain and investigate estimators independent of the chosen model.
Design is then as straightforward as choosing the embedding, deriving the correct model, and extracting the estimator.
We will demonstrate this by deriving novel estimators for time-varying systems and by relating two adjacent approaches for control Koopman operator learning~\cite{bevandaNonparametricControlKoopman2025,haseliModelingNonlinearControl2023, nuskeFiniteDataErrorBounds2023}.

\emph{Organization}
Section~II reviews the background on learning from data, inverse problems, and linear evolution equations.
In Section~III, we develop the structural theory of linear evolution equations for time-varying control systems and discuss special cases, including the Koopman operator.
Section~IV formulates the learning problem as an inverse problem and invokes the framework to perform error decompositions and derive estimators.
We conclude in Section~V.

\emph{Notation} We use $\Natural, \Whole, \Real$ for natural (with zero), integer, and real numbers, respectively.
We use $x\in\mathbb{X}$ to define elements of a domain that is often a subset of $\Real^n$.
Mappings between them are denoted by $f:\domX\to\domY$.
Function spaces are denoted by $\mathcal{F}(\mathbb{X};\domY)$, where we omit the image if it is clear from context.
Operators between function spaces are denoted by $\mathsf{L}:\mathcal{F}\to\mathcal{G}$.
In particular, we use $\Lp[p][\mu](\mathbb{X})$ for the space of $p$-integrable functions on $\mathbb{X}$ with the measure $\mu$. We omit the measure when it is clear from context and do not specify the $\sigma$-algebra, but state when separability is needed.
We use $\ell_p(\mathbb{X})$ for the sequences spaces with norm $\|\cdot\|_p$. The tensor product of two spaces is denoted by $\mathcal{F}\otimes\mathcal{G}$.
Hilbert-Schmidt operators are denoted by $\HS(\mathcal{F}; \mathcal{G})$.
We use $\langle f| \eval \rangle$ to denote a dual pairing, that is, the evaluation of a function $f\in \mathcal{F}:\domX\to\domY$ by a functional $\eval\in\mathcal{F}^*: \mathcal{F} \to \domY$.
When describing inverse problems, we will denote $\encoder: \mathcal{F}\to \outputspace$ the encoder producing data. The population encoder $\encoder$ has its range in the measure space $\outputspace$ equipped with a measure $\mu$, for example $\Lp[2][\mu]$. The sample encoder $\encoder_N: \mathcal{F}\to \domY^N$ has its range in the vector space $\domY^N$ and produces samples distributed according to measure $\mu$.
The sample consists of $N$ evaluations $\{\langle f|\eval_i\rangle\}_{i=1}^N$, where $\mu_N=\frac{1}{N}\sum_{i=1}^N\eval_i$ is the empirical measure.

\section{BACKGROUND}
To set the stage for our framework, we will first introduce the problem of learning from data and then review the theory of inverse problems, which will be the basis for our analysis of learning algorithms.
We will then introduce the theory of linear evolution equations, which will serve as the basis for our structural investigation of the models we aim to learn.
\subsection{Learning from Data}\label{section:LearningFromData}
We shall start by defining what we mean by building models from data.
\begin{problem}[Learning from Operational Data]
Let $\domY$ be a set of observed variables and $\domU$ a set of exogenous inputs.
Given $\data=\{y_i, u_i, y'_i\}_{i=1}^N$, find a model $\mathsf{M} \in \modelspace$ such that
\begin{align}\label{equation:GenericInverseProblem}
    \mathsf{E}(\mathsf{M})=\mathsf{T},
\end{align}
the target $\mathsf{T} \in \targetspace$ encodes the successors $\{y'_i\}$
and the encoder $\mathsf{E}: \modelspace \to \targetspace$ encodes model predictions
at the data $(y_i, u_i)$.
\end{problem}
Naturally, the answer to this question depends on the assumptions we are willing to make about the encoder, model, and target at hand.
And the art of modeling is in posing \eqref{equation:GenericInverseProblem} well.
We will omit the braces in this paper, as we restrict ourselves to the case where $\mathsf{E}$ is a linear operator between Hilbert spaces, encompassing commonly used approaches.
A popular and successful approach is to assume that a static function of the vectors $y$ and $u$ from a hypothesis space $\hypothesis$ describes the vector $y'$.
We summarize this as Ansatz~\ref{ansatz:EvolutionOfOutput}.
\begin{ansatz}[Evolution of Output]\label{ansatz:EvolutionOfOutput}
    The data is generated by evaluating the function that describes evolved observations in terms of past observations and inputs:
    \begin{align*}
        y' & = \langle h' | \eval_{y,u} \rangle
        ,
    \end{align*}
    where $\eval_{y,u}: \hypothesis \to \domY$ is a vector-valued evaluation functional and $h'\in\hypothesis:\domY\times\domU\to \domY$.
\end{ansatz}
Another approach is modeling the data generated by evaluating an observable and its evolved version, which we summarize as Ansatz~\ref{ansatz:EvolutionOfObservable}.
\begin{ansatz}[Evolution of an Observable]\label{ansatz:EvolutionOfObservable}
    The data is generated by evaluating an observable and its evolved version. The evolution is a linear operator $\mathsf{M}: \prehypothesis \to \hypothesis$:
    \begin{align*}
        y  & = \langle h| \eval_{y}\rangle    &
        h' & =  \mathsf{M}h                   &
        y' & = \langle h'| \eval_{y,u}\rangle
        ,
    \end{align*}
    where $\eval_y: \prehypothesis\to\domY$, $h\in\prehypothesis:\domY\to\domY$ and $h'\in\hypothesis:\domY\times\domU\to\domY$.
\end{ansatz}
A concrete example of $\eval$ is the function evaluation in a reproducing kernel Hilbert space $\eval_y(h)=\langle h, k(\cdot, y)\rangle$.
Ansatz~\ref{ansatz:EvolutionOfOutput} is generally called function learning, while Ansatz~\ref {ansatz:EvolutionOfObservable} is called operator learning. Both are visualized in Figure~\ref{fig:LearningComparison}. The striking difference is how evolution and evaluation are treated.
Ansatz~\ref{ansatz:EvolutionOfOutput} makes $h'$ the fundamental object responsible for evolution and observation.
In contrast, Ansatz~\ref{ansatz:EvolutionOfObservable} decouples the responsibility for evolution and observation, allowing for the definition of a pure evolution model in terms of the operator $\mathsf{M}$ and an observation model $h$.
\subsection{Hilbert Spaces}\label{section:HilbertSpaces}
We adopt the informative setting of the process living in a Hilbert space and being governed by a bounded linear operator.
While operator theory largely applies to more general Banach spaces~\cite{engelOneParameterSemigroupsLinear2000}, Hilbert spaces admit a complete inner product, a salient property that allows for orthogonal projections, the Pythagorean theorem, and Hilbert-Schmidt operators, making the analysis of learning algorithms more tractable.
Operators on Hilbert spaces model large classes of systems, i.e., Stochastic Differential Equations~\cite{kallenbergFoundationsModernProbability2021}, and when chosen appropriately can approximate continuous observables arbitrarily well~\cite{ingosteinwartSupportVectorMachines2008}.
An important class of Hilbert spaces is the reproducing kernel Hilbert space (RKHS). It admits continuous function evaluations, enabling learning and generalization from data.
\begin{definition}[RKHS~{\cite[Definition 2.2]{kanagawaGaussianProcessesReproducing2025}}]\label{def:RKHS}
    Let $\domY$ be a nonempty set and $k: \domY \times \domY \to \Real$ a positive definite kernel.
    A Hilbert space $\hypothesis_k$ of functions on $\domY$ with inner product $\langle \cdot, \cdot \rangle_{\hypothesis_k}$ is called a \emph{reproducing kernel Hilbert space} (RKHS) with reproducing kernel $k$ if:
    \begin{enumerate}
        \item For all $y \in \domY$, we have $k(\cdot, y) \in \hypothesis_k$.
        \item For all $y \in \domY$ and $f \in \hypothesis_k$,
              \begin{align}
                  f(y) = \langle f, k(\cdot, y) \rangle. \label{eq:ReproducingProperty}
              \end{align}
    \end{enumerate}
\end{definition}
\subsection{Inverse Problems}\label{section:InverseProblems}
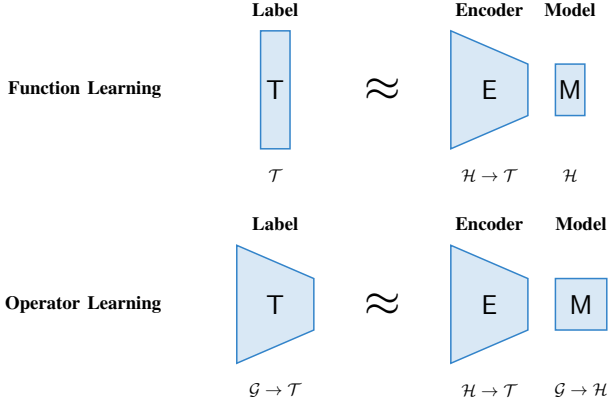
\begin{figure}[t]
    \centering
    \definecolor{TUMBlue}{RGB}{0, 101, 189}
    \def\latentsize{2.2cm}
    \def\latentsizecompact{0.96cm}
    \def\colgap{2.0}
    \def\rowgap{4.0}
    \def\headgap{1.5}
    \def\spacegap{1.6}
    \def\titlegap{2.0}
    \def\weightgap{0.5}
    \def\encysep{5pt}
    \def\encdepth{1.44cm}
    \pgfmathsetmacro{\encangle}{atan(4*\encysep/(\latentsize-\latentsizecompact))}
    \tikzset{
      ext/.style={trapezium, trapezium stretches body, trapezium angle=\encangle,
      draw=TUMBlue!80, fill=TUMBlue!15, inner ysep=\encysep, outer sep=0pt,
      minimum height=\encdepth, minimum width=\latentsize},
      latent/.style={rectangle, draw=TUMBlue!80, fill=TUMBlue!15,
      minimum height=\latentsizecompact, minimum width=\latentsizecompact},
      narrowrect/.style={rectangle, draw=TUMBlue!80, fill=TUMBlue!15,
      minimum height=\latentsize, minimum width=0.55cm},
      narrowrectcompact/.style={narrowrect, minimum height=\latentsizecompact},
      nextcol/.style={at={($(#1.center)+(\colgap,0)$)}},
      nextrow/.style={at={($(#1.center)+(0,-\rowgap)$)}},
      headerof/.style={align=center, at={($(#1.center)+(0,\headgap)$)}},
      spaceof/.style={font=\small, text=black, at={($(#1.center)+(0,-\spacegap)$)}},
      titleof/.style={align=right, anchor=east, at={($(#1.center)-(\titlegap,0)$)}},
      weightof/.style={anchor=west, at={($(#1.north)+(\weightgap,0)$)}},
    }
    \resizebox{\columnwidth}{!}{%
    \begin{tikzpicture}[thick]
      \node [narrowrect] (F_target) at (0,0) {};
      \node at (F_target) {\Large$\mathsf{T}$};
      \node [titleof=F_target] {\textbf{Function Learning}};
      \node [headerof=F_target] {\textbf{Label}};
      \node [spaceof=F_target] {$\mathcal{T}$};
      \node [nextcol=F_target] (F_approx) {\Huge$\approx$};
      \node [ext, rotate=-90, nextcol=F_approx] (F_enc) {};
      \node at (F_enc) {\Large$\mathsf{E}$};
      \node [narrowrectcompact, weightof=F_enc] (F_weights) {};
      \node at (F_weights) {\Large$\mathsf{M}$};
      \node [headerof=F_enc] {\textbf{Encoder}};
      \node [headerof=F_weights] {\textbf{Model}};
      \node [spaceof=F_enc] {$\mathcal{H}\to\mathcal{T}$};
      \node [spaceof=F_weights] {$\mathcal{H}$};
      \node [ext, rotate=-90, nextrow=F_target] (O_encL) {};
      \node at (O_encL) {\Large$\mathsf{T}$};
      \node [titleof=O_encL] {\textbf{Operator Learning}};
      \node [headerof=O_encL] {\textbf{Label}};
      \node [spaceof=O_encL] {$\mathcal{G}\to\mathcal{T}$};
      \node [nextcol=O_encL] (O_approx) {\Huge$\approx$};
      \node [ext, rotate=-90, nextcol=O_approx] (O_encR) {};
      \node at (O_encR) {\Large$\mathsf{E}$};
      \node [latent, weightof=O_encR] (O_weights) {};
      \node at (O_weights) {\Large$\mathsf{M}$};
      \node [headerof=O_encR] {\textbf{Encoder}};
      \node [headerof=O_weights] {\textbf{Model}};
      \node [spaceof=O_encR] {$\mathcal{H}\to\mathcal{T}$};
      \node [spaceof=O_weights] {$\mathcal{G}\to\mathcal{H}$};
      \path ([shift={(-10pt,-10pt)}]current bounding box.south west)
        rectangle ([shift={(10pt,10pt)}]current bounding box.north east);
    \end{tikzpicture}}%
    \caption{Function learning versus operator learning. \textbf{Top:} The target $\target$ is a vector of function evaluations. The model $\model$ is a function and is evaluated via the encoder $\encoder$. \textbf{Bottom:} The target $\target$ is the evaluation of the true operator's action when restricted to $\mathcal{G}$. The operator $\model$ models the action on functions in $\prehypothesis\rightarrow\hypothesis$ which are evaluated via the encoder $\encoder$.}
    \label{fig:LearningComparison}
\end{figure}
This exposition to the basic questions studied in inverse problems and theory follows the seminal work of~\cite{englRegularizationInverseProblems1996}.
We focus on the basic questions and the answers relevant to later sections and refer to~\cite{englRegularizationInverseProblems1996} for a complete coverage.
\begin{problem}[Inverse Problem~{\cite[p.31]{englRegularizationInverseProblems1996}}]\label{problem:InverseProblem}
Let the encoder $\encoder$ be a bounded operator between Hilbert spaces $\modelspace$ and $\targetspace$. Given the target $\target\in\targetspace$ find a model $\model\in\modelspace$ such that
\begin{align}\label{eq:InverseProblem}
    \encoder \model = \target
    .
\end{align}
This problem is called well-posed if the following hold.
\begin{itemize}
    \item[(E)] For all admissible data, a solution exists.
    \item[(U)] For all admissible data, the solution is unique.
    \item[(C)] The solution depends continuously on the data.
\end{itemize}
\end{problem}
The requirement (E) is usually addressed by ensuring the model class is sufficiently large or by relaxing the strict condition in \eqref{eq:InverseProblem}, which requires the encoder to be continuously invertible. One sensible way is to settle on the solution with the best data fit, the least-squares solution: $\inf_{\model\in\modelspace}\|\encoder \model-\target\|^2$. In case of multiple least-squares solutions, (U) can then be ensured by choosing the solution with the smallest norm, the so-called best-approximate solution $\|\model\|=\inf\{\|\tilde{\model}\| \mid \tilde{\model}\text{~is a least-squares solution}\}$.

The Moore-Penrose (MP) inverse of the encoder $\encoder^\dagger: \targetspace\to\modelspace$~\cite[Definition 2.2]{englRegularizationInverseProblems1996}, can be used to obtain the best-approximate solution $\model^\dagger:=\encoder^\dagger\target$. If $\target\in\text{domain}(\encoder^\dagger)$ the MP inverse operator is defined through the normal equations $\encoder^*\encoder \model=\encoder^*\target\implies\encoder^\dagger=(\encoder^*\encoder)^\dagger\encoder^*$ which are usually easier to solve. However, even for the normal equations, $\encoder^\dagger$ does not continuously depend on data~\cite{englRegularizationInverseProblems1996}.
This is because when $\encoder^\dagger$ is unbounded, an arbitrarily small deviation in the data $\|y^\deviation-y\|\leq\deviation$ can lead to arbitrary changes in the model.

The remedy is to introduce a regularization that ensures continuity while converging to the best-approximating solution as the level of deviation decreases.
\begin{definition}[Regularization~{\cite[Definition~3.1]{englRegularizationInverseProblems1996}}]\label{definition:RegularizationOperator}
    For an inverse problem defined in Problem~\ref{problem:InverseProblem}, let $\gamma_0\in (0,\infty]$. For every $\gamma\in (0,\gamma_0)$, let $\regularization_\gamma:\targetspace\to\modelspace$ be a continuous (not necessarily linear) operator. The family $\{\regularization_\gamma\}$ is called a regularization, if, for all $\target\in \operatorname{domain}(\encoder^\dagger)$, there exists a parameter choice rule $\alpha(\epsilon, \target^\epsilon)$ such that
    \begin{equation*}
    \ensuremath{\lim_{\epsilon\to 0}\sup_{\target^\epsilon\in\targetspace} \{\|\regularization_{\alpha (\epsilon, \target^\epsilon)}\target^\epsilon-\encoder^\dagger\target\|,\text{~where~}\|\target^\epsilon-\target\|\leq\epsilon\}=0.}
    \end{equation*}
\end{definition}

Both the unboundedness issues of the best-approximation and important regularization methods can be understood by calculating them in terms of the spectrum of $\encoder^*\encoder=\int \lambda dP_\lambda$:
$\model^\dagger=(\encoder^*\encoder)^{-1}\encoder^*\target=\int\frac{1}{\lambda}dP_\lambda \encoder^*\target$, where $P_\lambda$ is a spectral projection.
One immediately observes that if $\encoder^*\encoder$ is rank deficient, the MP inverse is unbounded.
For operators on infinite-dimensional spaces, the same issue arises due to compactness. It is thus a natural approach to define regularizers in terms of the spectrum:
$\model^{\gamma}=\int g_\gamma({\lambda})dP_\lambda \encoder^*\target$. This approach encompasses two commonly used regularizers: Tikhonov, $g^{\text{Tikhonov}}_\gamma(\lambda)=\frac{1}{\lambda+\gamma}$, and principal components~\cite{kosticLearningDynamicalSystems2022}, $g^{\text{PC}}_\gamma(\lambda)=\{
    \frac{1}{\lambda} \text{~if~}\lambda {\geq} \gamma, \text{else~}0\}$. A variant of $g^{\text{PC}}_\gamma$ is also used in the original Extended Dynamic Mode Decomposition (EDMD)~\cite{williamsDataDrivenApproximationKoopman2015,o.williamsKernelbasedMethodDatadriven2015}.
\subsection{Linear Evolution Equations}
This section introduces the theory of linear evolution equations, forming the basis for our structural investigation of the models.
Based on this framework, one can choose which objects to learn, which to analyze, and how to relate them.
We follow the seminal work of Engel and Nagel~\cite{engelOneParameterSemigroupsLinear2000}, which provides a comprehensive theory of the topic.
\begin{definition}[Semigroup~{\cite{engelOneParameterSemigroupsLinear2000}}]\label{def:Semigroup}
    Let $\mathcal{F}$ be a Hilbert space.
    A family of bounded linear operators $\{\mathsf{S}(t)\}_{t \geq 0}$ on $\mathcal{F}$ is called a semigroup if it satisfies the following properties:
    \begin{itemize}
        \item $\mathsf{S}(0) = \identity$, where $\identity$ is the identity operator on $\mathcal{F}$.
        \item $\mathsf{S}(t + s) = \mathsf{S}(t)\mathsf{S}(s)$ for all $t, s \geq 0$.
        \item $t \mapsto \mathsf{S}(t)f$ is continuous for each fixed $f \in \mathcal{F}$.
    \end{itemize}
\end{definition}
We will restrict our attention to the semigroups $t,s\in \Natural$ that resemble the discrete-time evolution of a system.
We note that in this case the semigroup property implies that $\mathsf{S}(t) = \mathsf{S}(1)^t$ for all $t \in \Natural$, so the semigroup's generator is $\mathsf{A}:=\mathsf{S}(1)$.
\begin{definition}[Evolution Families~{\cite[Definition 9.2]{engelOneParameterSemigroupsLinear2000}}]\label{def:EvolutionFamily}
    Let $\mathcal{F}$ be a Hilbert space.
    A family of bounded linear operators $\{\mathsf{U}(t,s)\}_{t \geq s \in\Whole}$ on $\mathcal{F}$ is called an evolution family if it satisfies the following properties:
    \begin{enumerate}
        \item $\mathsf{U}(s,s) = \identity$ for all $s \geq 0$.
        \item $\mathsf{U}(t,r) = \mathsf{U}(t,s)\mathsf{U}(s,r)$ for all $t \geq s \geq r \geq 0$.
        \item $(t,s) \mapsto \mathsf{U}(t,s)f$ is continuous for each fixed $f \in \mathcal{F}$.
    \end{enumerate}
\end{definition}
In analogy to semigroups, we can define the generator as $\mathsf{A}(t):= \mathsf{U}(t+1, t)$, which, through property 2, determines the evolution family.
\begin{lemma}[Evolution Semigroup~{\cite[Lemma 9.10]{engelOneParameterSemigroupsLinear2000}}]\label{lemma:EvolutionSemigroup}
    Equation~\eqref{eq:EvolutionSemigroup} defines a semigroup $\{\mathsf{G}(t)\}_{t \geq 0}$ on sequences $\xi\in\sequencespace:=\ell_\infty(\Whole, \mathcal{F})$, called the evolution semigroup.
    \begin{align}\label{eq:EvolutionSemigroup}
        [\mathsf{G}(t)\xi](s) =\mathsf{U}(s, s-t)\xi (s-t)
    \end{align}
    The generator of the evolution semigroup is defined in terms of the generator of the evolution family and the right shift operator $\mathsf{R}:\sequencespace\to\sequencespace$ as $\mathsf{A}_{\text{e}}:=\mathsf{G}(1) =\mathsf{A}(\cdot-1)\mathsf{R}$.
\end{lemma}
In summary, the framework of semigroups and evolution families allows capturing the structure of evolution equations at different levels of abstraction.
While we introduce descriptions for time-varying evolutions, we will see that control systems are a special parametrization of this class.
\begin{remark}\label{remark:ContinuousTime}
    Although we restrict ourselves to the discrete-time case, the semigroup and evolution family framework naturally applies to continuous-time systems.
    In this paper, we call $\mathsf{S}(1)$ the generator of the semigroup, in the straightforward group theoretic sense.
    In continuous time where the evolution operators are defined for all $t \in\Real_{0}^{+}$, the situation is more delicate, as no $\mathsf{S}(t_0)$ characterizes the semigroup.
    Instead, the generator of the semigroup is then given by the infinitesimal limit $\mathsf{A}=\lim_{h \to 0} \frac{\mathsf{S}(h) - I}{h}$.
\end{remark}
\section{Linear Time-Varying Evolution Equations}\label{section:LinearEvolutionEquationsForControl}
Motivated by the transition map of a dynamical system, we will derive the structure for control systems modeled by composition operators.
\subsection{General Structure}
\begin{definition}[Dynamical System]\label{def:DynamicalSystem}
    Let $\domY$ be a domain and $\varphi: \domY \to \domY$ the transition map on said domain. Defining the dynamical system $y_{t+1}=\varphi(y_t)$. Let $\varphi$ be non-singular and absolutely continuous, meaning $\mu(a)=0\implies\mu\circ \varphi^{-1}(a)=0$, and $\mu\circ \varphi^{-1}\ll\mu$.
\end{definition}
These conditions imply that sets of measure zero still have measure zero under the pullback -- excluding dynamical behavior like finite-time blow-up and collapse.

We will now change our view of the evolution from the vector being fed through the transition map to the direct evolution of the observation.
To this end, we impose that the observation comes from an observable $h\in\hypothesis$ evaluated via the evaluation functional $\eval$.
The composition with the transition map $\mathsf{K}: \hypothesis \to \hypothesis$ as defined in~\eqref{eq:EvolutionOfObservable} is the generator of the evolution semigroup.
\begin{align}\label{eq:EvolutionOfObservable}
    y_{t+1} = \langle h| \eval_y\rangle = h(\varphi(y_t)) = [\mathsf{K}h](y_t)=[\mathsf{S}(1)h](y_t)
    .
\end{align}
\begin{proposition}[Semigroup of a Dynamical System]\label{proposition:Semigroup}
    The operator $\mathsf{S}(1)$ defined in~\eqref{eq:EvolutionOfObservable} is a bounded linear operator on $\Lp(\domY)\to \Lp(\domY)$.
    It is the generator of a semigroup $\{\mathsf{S}(t)\}_{t \in \Natural}$ of evolution operators.

    \begin{proof}
        The operator $\mathsf{S}(1)$ is linear as it is a composition.
        Boundedness follows from non-singularity and $\mu\circ \varphi^{-1}\ll \mu$, as this implies that $\mathsf{S}(1)$ maps open sets to open sets~\cite[Definition II.1.1]{wernerFunktionalanalysis2018},\cite{stochelCompositionOperatorsL2spaces}.
        The semigroup property follows from $\mathsf{S}(s+t) h = h\underbrace{\circ \varphi \circ \cdots \circ \varphi}_{s+t~times} =\mathsf{S}(s)\mathsf{S}(t)h$ for all $s,t \in \Natural$ and by defining $\mathsf{S}(0) = \identity$, which is bounded in $\Lp$.
    \end{proof}
\end{proposition}
\begin{remark}
    If $\varphi^{-1}$ defines a semigroup in the same way as $\varphi$, then they can be combined to a \emph{group} of evolution operators.
    This is the case for invertible systems, but not for general forward complete systems.
\end{remark}
We will now perform an analogous construction for time-varying and control systems, and verify that autonomous systems are a special case of this construction.
\begin{definition}[Time-Varying System]\label{def:TVSystem}
    Let $\domY$ be the domain of the transition map $\varphi: \domT \times\domY \to \domY$ defining a dynamical system via $y_{t+1} = \varphi(t, y_t)$.
    Let $\forall t: \varphi(t,\cdot)$ be non-singular and absolutely continuous.
\end{definition}
We can again shift our view of the evolution from the transition to the evolution of the observation.
\begin{align}\label{eq:TVSystem}
    y_{t+1} = \langle h| \eval_y\rangle  = h(\varphi(t, y_t)) = [\mathsf{K}h](t, y_t)
    .
\end{align}
In this case, the composition operator $\mathsf{K}$ maps observables from $\Lp(\domY)$ to $\Lp(\domY \times \domT)$ via composition with the time-varying transition map $\varphi$. We observe that the natural time-varying extension of the evolution equations of composition operators does not admit the semigroup property $S(t+s)=S(t)S(s)$. Yet, by defining $\mathsf{G}(t)f=[\mathsf{K}f](t, \cdot)$ we can construct an evolution family as an instance of Definition~\ref{def:EvolutionFamily}.
\begin{lemma}[Evolution Family of a Time-Varying System]\label{lemma:EFamilyTVSystem}
    The operators $\mathsf{G}(t)=:\mathsf{U}(t+1,t)$ generate an evolution family $\{\mathsf{U}(t, s)\}_{t\geq s \in \Whole}$ of evolution operators on $\Lp(\domY)$.
    \begin{proofsketch}
        The proof immediately follows from the definitions by substituting~\eqref{eq:TVSystem} in Definition~\ref{def:EvolutionFamily}.
    \end{proofsketch}
\end{lemma}
\begin{lemma}[Time-Varying Evolution Semigroup]\label{lemma:TVSystemSemigroup}
    Let $\xi \in \sequencespace=\ell_\infty(\Whole, \Lp(\domY))$  be the sequence observable on the extended space of sequences of observables.
    The Evolution Family from Lemma~\ref{lemma:EFamilyTVSystem} induces a evolution semigroup on $\sequencespace$ defined as
    \begin{align}
        [\mathsf{S}(t)\xi](s) & =\mathsf{U}(s, s-t)\xi (s-t)                                                                                \\
                              & =\mathsf{G}(s-1)\mathsf{G}(s-2)\cdots \mathsf{G}(s-t)\xi (s-t)                                              \\
                              & =h_{s-t}\circ \varphi(s-1, \cdot)\circ \cdots \circ \varphi(s-t, \cdot)\label{eq:TVEvoSemigroupSubstituted}
        .
    \end{align}
    \begin{proofsketch}
        The proof immediately follows from the definitions by substituting~\eqref{eq:TVEvoSemigroupSubstituted}.
    \end{proofsketch}
\end{lemma}
The evolution semigroup on the extended space of sequences of observables recovers the semigroup structure.
While this allows for analysis, it is not a convenient structure for learning:
When learning the evolution semigroup $\mathsf{S}(t)$, we run into the problem that $\sequencespace$ is not a Hilbert space and is not separable. When learning the composition operator $\mathsf{K}:\Lp(\domY) \to \Lp(\domY \times \domT)$, we run into the issue that $\domT$ is not compact.
It is not surprising that those spaces pose issues for learning, as they describe the evolution for all times directly.
A possible remedy is to restrict attention to a finite time interval $[0, T]\subset\domT$, modeling periodic behavior or transients by simply neglecting long-term behavior.

If, however, we have more information about the time dependence of the system and can attribute it to exogenous variables $u(t)\in\domU$, we can uncover a more convenient structure of the evolution operator $\mathsf{G}(t)$.
\begin{definition}[Control System]\label{def:ControlSystem}
    Let $\domY$ be a domain, $\domU$ be a control domain. The transition map ${\varphi}: \domY \times \domU \to \domY$ defines a dynamical system $y_{t+1} = {\varphi}(y_t, u_t)$. Let $\varphi$ be non-singular and absolutely continuous.
\end{definition}
We define the evolution of the observation in analogy to \eqref{eq:TVSystem}, identifying $t$ with $u_t$:
\begin{align}\label{eq:ControlledObservableEvolution}
    y_{t+1} = \langle h| \eval_{y,u}\rangle = h(\varphi(y_t, u_t)) = [\mathsf{K}h](y_t, u_t)
    .
\end{align}
In this case, the operator $\mathsf{K}$ maps observables from $\Lp(\domY)$ to $\Lp(\domY \times \domU)$ via composition with the control transition map $\varphi$.
Obviously, this operator does not yet admit the semigroup or evolution family structure.
The next statements show that with an additional composition, we recover said structure.
\begin{definition}[Control Koopman Operator~{\cite[Def.~1]{bevandaNonparametricControlKoopman2025}}]\label{definition:ControlKoopmanOperator}
    The control Koopman operator $\mathsf{K}: \Lp(\domY) \to \Lp(\domY \times \domU)$ is defined as $[\mathsf{K}h](y, u) = h(\varphi(y, u))$.
\end{definition}
\begin{definition}[Control Operator]\label{definition:ControlOperator}
    Define the partial composition operator $\mathsf{C}(t): \Lp(\domY \times \domU) \to \Lp(\domY)$ as $[\mathsf{C}(t) g](y) = g(y, u(t))$, where $u(t):\domT \to \domU$.
\end{definition}
\begin{proposition}[Control Evolution Operator]\label{prop:EvolutionOperatorControlSystem}
    In analogy to Definition~\ref{def:DynamicalSystem}, the control evolution operators $\mathsf{G}(t)=\mathsf{C}(t)\mathsf{K}: \Lp(\domY) \to \Lp(\domY)$ are the generators of an evolution family $\{\mathsf{U}(t,s)\}_{t,s\in \Whole}$ of evolution operators.

    \begin{proof}
        Linearity and boundedness of $\mathsf{G}$ follow as in Proposition~\ref{proposition:Semigroup}.
        Define $\mathsf{U}(t,s) = \mathsf{G}(t-1)\mathsf{G}(t-2)\cdots \mathsf{G}(s)$ for $t > s$ and $\mathsf{U}(t,t) = \identity$.
        Then $\mathsf{U}(t,s) = \mathsf{U}(t,r)\mathsf{U}(r,s)$ for $t \geq r \geq s$ by splitting the product.
    \end{proof}
\end{proposition}
\begin{proposition}[Control Evolution Semigroup]
    Let $\xi \in \ell_\infty(\Whole, \Lp(\domY))$  be the sequence observable on the extended space of sequences of observables.
    Then the Control Evolution Family from Proposition~\ref{prop:EvolutionOperatorControlSystem} induces a evolution semigroup on $\ell_\infty(\Whole, \Lp(\domY))$ defined as
    \begin{align}
        [\mathsf{S}(t)\xi](s) =\mathsf{U}(s, s-t)\xi (s-t)
        .
    \end{align}

    \begin{proof}
        The identity element is $\mathsf{S}(0)$, as by definition $[\mathsf{S}(0)\xi](s)=\mathsf{U}(s,s)\xi(s)=\xi(s)$.
        The semigroup property follows from $[\mathsf{S}(t+r)\xi](s) = \mathsf{U}(s, s-t-r)\xi(s-t-r) = \mathsf{U}(s, s-t)\mathsf{U}(s-t, s-t-r)\xi(s-t-r) = [\mathsf{S}(t)\mathsf{S}(r)\xi](s)$ for all $t,r,s \in \Whole$.
    \end{proof}
\end{proposition}
\subsection{Structure Of Special Cases}
\begin{example}[Semigroup of a Linear Control System]\label{example:LinearControlSystem}
    Let $\domY = \Real$, $\domU = \Real$, and $\varphi(y, u) = y + u$, defining the control system $y_{t+1} = ay_t + bu_t$.
    On the observable $h: y \mapsto y$, the control Koopman operator acts as $[\mathsf{K}h](y, u) = ay + bu$.
    The action of the control evolution semigroup $\mathsf{S}(t)$ on a time-uniform observable sequence $\xi = (h)_{s\in\Whole}\in\ell_\infty(\Whole, L_2(\domY))$ is $\mathsf{S}(1)\xi = (ay + bu_{s-1})_{s\in\Whole}$, a sequence of functions describing how $y$ was evolved from the previous time step.
    Applying $\mathsf{S}(t)$ yields $\mathsf{S}(t)\xi = (a^t y + \sum_{i=1}^t a^{i-1}bu_{s-i})_{s\in\Whole}$, describing how $y$ was evolved from $t$ time steps ago.
\end{example}
\begin{example}[Autonomous Systems]
    For autonomous systems $x_{t+1} = \varphi(x_t)$ as in Definition~\ref{def:DynamicalSystem}, there is no control input, so the control Koopman operator $\mathsf{G}$ reduces to the Koopman operator $\mathsf{S}(1): \Lp(\domX) \to \Lp(\domX)$ with $[\mathsf{S}(1)h](x) = h(\varphi(x))$, the partial composition operator $\mathsf{C}(t)$ becomes the identity, and the evolution operators $\mathsf{G}(t) = \mathsf{S}(1)$ are time-invariant.
    The evolution family thus satisfies $\mathsf{U}(t,s) = \mathsf{S}(1)^{t-s}$, depending only on the difference $t-s$, recovering the semigroup $\{\mathsf{S}(1)^t\}_{t\in\Natural}$ from Proposition~\ref{proposition:Semigroup}.
\end{example}
\begin{example}[Time-Varying Systems]
    Let $\varphi: \Whole\times\domY \to \domY$ be a $\mu${-}measurable time-varying transition map defining a time-varying system via $y_{t+1} = \varphi(t, y_t)$.
    On the observable $h \in L_2(\domY)$ the evolution operator at time $t$ acts as $[\mathsf{G}(t)h](y) = h(\varphi_t(y))$.
    The action of the evolution semigroup $\mathsf{S}(t)$ on a time-uniform observable sequence $\xi = (h)_{s\in\Whole}\in\ell_\infty(\Whole, L_2(\domY))$ is $\mathsf{S}(1)\xi = (h \circ \varphi_{s-1})_{s\in\Whole}$, a sequence of functions describing how $h$ was evolved from the previous time step.
    Applying $\mathsf{S}(t)$ yields $\mathsf{S}(t)\xi = (h \circ \varphi_{s-t} \circ \cdots \circ \varphi_{s-1})_{s\in\Whole}$, describing how $h$ was evolved through $t$ successive composition with the time-varying transition map.
\end{example}
\section{LEARNING LINEAR EVOLUTION EQUATIONS AS AN INVERSE PROBLEM}
Coming back to the modeling approaches given by Ansatz~\ref{ansatz:EvolutionOfOutput} and Ansatz~\ref{ansatz:EvolutionOfObservable}, we look at the corresponding Inverse Problem~\ref{problem:InverseProblem} they instantiate:
\begin{problem}[Output Inverse Problem]\label{problem:StateSpaceInverse}
With Ansatz~\ref{ansatz:EvolutionOfOutput}, the model consists of the evolved observation function $\mathsf{M} = h' \in \hypothesis: \domY \times \domU \to \domY$, the encoder is the evaluation functional $\encoder = \langle \cdot | \eval_{y,u} \rangle: \hypothesis \to \domY$, and the target represents the successors $\mathsf{T} = y'$.
Find $h'$ such that
    \ensuremath{\encoder h' = y'.}
\end{problem}
\begin{problem}[Operator Inverse Problem]\label{problem:OperatorInverse}
Under Ansatz~\ref{ansatz:EvolutionOfObservable}, the model is the evolution operator $\mathsf{M}:\prehypothesis\to\hypothesis$, the encoder is the evaluation functional $\mathsf{E} = \langle \cdot | \eval_{y,u} \rangle: \hypothesis \to \outputspace$, and the target maps observables to their successor values $\mathsf{T}: \prehypothesis\to\outputspace$.
Find $\mathsf{M}$ such that
    \ensuremath{\mathsf{E}\mathsf{M} = \mathsf{T}.}
\end{problem}
Note that the encoder doubles as $\encoder: \hypothesis \to \outputspace$ and left composition $\encoder(\model)=\encoder\model$, leading to  $\encoder: \modelspace = \HS(\prehypothesis, \hypothesis) \to \targetspace = \HS(\prehypothesis, \outputspace)$, which connects Problem~\ref{problem:OperatorInverse} to the abstract framework of Problem~\ref{problem:InverseProblem}.

By comparing the inverse problems, we can identify three advantages of the operator learning Ansatz~\ref{ansatz:EvolutionOfObservable} over the function learning Ansatz~\ref{ansatz:EvolutionOfOutput}:
\begin{itemize}
    \item It is more general, as the evaluation is not specified, and different modes of evaluation are encompassed.
          Two important examples are the point evaluation $\encoder h=\eval_{y}(h)=h(y)$ and the expectation of an observable $\encoder h=\eval_{\mu}(h) =\mathbb{E}_{Y\sim\mu}h(Y)$ - both are linear operators, and allow for the same model.
    \item The target space $\targetspace$ naturally models both population-level quantities $\Lp[2][\mu](\domY)$ as well as samples $y_i\in\domY^N$, allowing for the construction of algorithms that are well defined in the data limit and thus for any dataset.
    \item The linear evolution structure simplifies analysis.
\end{itemize}
These advantages come at the cost of needing to represent elements of one additional hypothesis space $\prehypothesis\subset \Lp[2][\mu](\domX)$.

\subsection{Inverse Problems For Operators}
The main idea for solving Problem~\ref{problem:OperatorInverse} in a unified way is recognizing that, under reasonable conditions, the model and target spaces are Hilbert-Schmidt -- thus elements of Hilbert spaces themselves.
This makes the rich theory of inverse problems apply.
The first condition is on the data-generating process. It immediately applies to the operators on non-sequence spaces from previous sections.
\begin{condition}[Generating Process]\label{condition:GeneratingProcess}
    Let $\world: \Lp[2][\mu] \to \Lp[2][\nu]$ be a bounded linear operator relating the populations: the vectors in the measure spaces from which samples are drawn.
\end{condition}
The second condition ensures that $\prehypothesis$ is small enough.
\begin{condition}[Hypothesis]\label{condition:Hypothesis}
    Let $\prehypothesis, \hypothesis$ be Hilbert spaces, such that $\encodertarget\in \HS(\prehypothesis, \Lp[2][\nu])$.
\end{condition}
This can be ensured by letting $\prehypothesis$ be finite-dimensional or a universal RKHS.
Together, these requirements allow us to specify the model spaces.
\begin{proposition}[$\HS$-Model]\label{proposition:HSModel}
    Let Conditions~\ref{condition:GeneratingProcess}\&\ref{condition:Hypothesis} hold. Each linearly regularized model is Hilbert-Schmidt.

    \begin{proof}
        As the regularization $\regularization_\gamma: (\prehypothesis\to\targetspace)\to\modelspace$ is linear and bounded, we can represent the model as $\model_R:= \regularization_\gamma \target=\regularization_\gamma \world\encodertarget$. The statement follows, as $\encodertarget$ is Hilbert-Schmidt and Hilbert-Schmidt operators are ideals in the bounded operators~\cite[Proposition 12.1.2]{aubinAppliedFunctionalAnalysis2000}.
    \end{proof}
\end{proposition}

\begin{remark}[Representation Error]\label{remark:RepresentationError}
    Condition~\ref{condition:Hypothesis} favors a small hypothesis space. Consequently, choosing a very small hypothesis $\prehypothesis$ simplifies estimation. However, this comes at a cost when using $\model$ as a surrogate for $\world$, for example, for spectral analysis or forecasting: Using the target $\target=\world\encodertarget:\Lp[2][\mu]\supset\prehypothesis\to \Lp[2][\nu]$ instead of $\world:\Lp[2][\mu]\to \Lp[2][\nu]$ produces a representation error $\|\target - \world\|$. This dichotomy has been identified in~\cite{kosticSharpSpectralRates2023} and commands the search for `good' representations.
\end{remark}
\begin{remark}[Hilbert-Schmidt Target Operators]\label{remark:HilbertSchmidtTarget}
    If $\mathsf{T}$ is not Hilbert-Schmidt, more sophisticated estimates are needed. They lead to weaker convergence~\cite{mollenhauerStatisticalApproximationConditional2022} and are at the forefront of statistical learning.
\end{remark}
\subsection{Error Decomposition}\label{section:ErrorDecomposition}
We will be interested in the quality of the estimators for $\model$. To this end, we measure the error $\error(\model )=\|\encoder\model-\world\|$.
To identify the components of this error, we decompose it into parts arising from the representation, inaccurate targets $\target^\epsilon$, associated regularization $\regularization_\gamma$, finite data $N$, and the hypothesis on finite subspaces $\hypothesis_m\subset\hypothesis$.
Denote by $\model_m^\dagger\in\mathrm{HS}(\prehypothesis, \hypothesis_m)$ the best-approximate, by $\regularization_{\gamma,m}: \HS (\prehypothesis, \hypothesis_m) \to \targetspace$ the regularization and by $\model_{\gamma,m,N}^\epsilon$ the regularized estimator from $N$ samples with deviation level $\epsilon$.
\begin{proposition}[Error Decomposition]\label{prop:ErrorDecomposition}
    Let $c>0$. The total error of the estimator $\model_{\gamma,m,N}^\epsilon$ satisfies
    \begin{align}\label{eq:ErrorDecomposition}
        \error(\model_{\gamma,m,N}^\epsilon)
         & \leq \underbrace{\|\target-\world\|}_{\text{representation}}
        + \underbrace{\|\encoder\model^\dagger-\target\|}_{\text{least squares}} \nonumber                             \\
         & + \underbrace{\|\model^\dagger - \model_m^\dagger\|}_{\text{approximation}}
        + \underbrace{\|\model_m^\dagger - \regularization_{\gamma,m}\target\|}_{\text{regularization bias}} \nonumber \\
         & + \underbrace{\|\regularization_{\gamma,m}\| \cdot \epsilon}_{\text{aleatoric}}
        + \underbrace{\|\regularization_{\gamma,m}\target^\epsilon - \model_{\gamma,m,N}^\epsilon\|}_{\text{estimation}}
        .
    \end{align}

    \begin{proofsketch}
        The proof follows by adding and subtracting intermediate terms and repeatedly applying the triangle inequality. The move from an error in target space to an error in model space is $\error(\model )=\|\encoder\model-\world\|=\|\target-\world + \encoder\model^\dagger-\target + \encoder(\model - \model^\dagger)\|$. As $\encoder$ is bounded, we set $\|\encoder\|=1$. This immediately gives the representation and least-squares terms. The remaining terms follow immediately by using $\|\mathsf{T}-\mathsf{T}^\epsilon\|\leq\epsilon$. The norms are operator norms and are upper-bounded by the Hilbert-Schmidt norms if finite.
    \end{proofsketch}
\end{proposition}
The specific error decomposition is not unique; nonetheless, the decomposition above highlights the existing bias-variance trade-offs.
The first trade-off is between the representation error and the least-squares error.
Next, the second trade-off involves the approximation error and the remaining terms.
Finally, the third trade-off is between the regularization bias and the aleatoric term.
The error estimates can be used to derive optimal schedules $\gamma=\alpha(\epsilon, \target^\epsilon)$ that make both terms vanish, when $\epsilon\to 0$.
In practice, however, this never happens as numerical precision gives a lower bound on $\epsilon$.
\subsection{Estimators For Special Cases}
\begin{example}[Control Koopman Operator]\label{example:CKOEstimator}
    Let $\hypothesis = \hypothesis(\domY)\otimes\hypothesis(\domU)$ and $\prehypothesis = \hypothesis(\domY)$ be separable RKHSs.
    The model is the control Koopman operator $\model = \mathsf{K}: \prehypothesis \to \hypothesis$ from Definition~\ref{definition:ControlKoopmanOperator},
    the encoder is the evaluation functional $\encoder = \langle \cdot | \eval_{y,u} \rangle: \hypothesis \to \outputspace$,
    and the target maps observables to their successor values $\target = \mathsf{K}\encodertarget: \prehypothesis \to \outputspace$.
    Given data $\data = \{y_i, u_i, y'_i\}_{i=1}^N$,
    the sample encoder is $\encoder_N: \hypothesis \to \domY^N$, $[\encoder_N h]_i = h(y_i, u_i)$,
    and the sample target is $\target_N \in \domY^N$, $[\target_N]_i = y'_i$.
    Learning $\model$ is an instance of Problem~\ref{problem:OperatorInverse}.
    The Tikhonov regularized best-approximate solution is given by
    \begin{align}\label{eq:EstimatorCKO}
        \model & = (\encoder_N^*\encoder_N + \gamma \identity_N)^{-1}\encoder_N^*\target_N                                      \\
               & = \encoder_N^* (\encoder_N\encoder_N^* + \gamma \identity_N)^{-1} \target_N\label{eq:EstimatorCKO:GramianForm}
        ,
    \end{align}
    using the Woodbury identity to obtain~\eqref{eq:EstimatorCKO:GramianForm}.

    The control operator $\mathsf{C}(t)$ does not need to be learned. By separability of $\hypothesis$ it is $\mathsf{C}(t)=\identity\otimes\langle\cdot|\eval_{u_t}\rangle$.
    Separability is necessary for statistical guarantees, but does not prevent $\hypothesis(\domY)\otimes\hypothesis(\domU)$ from being universal for continuous functions on $C(\domY\times\domU)$~\cite{ingosteinwartSupportVectorMachines2008}.
    As the separable encoder reads $\encoder_N=\encoder_{N, \hypothesis(\domY)}\otimes\encoder_{N, \hypothesis(\domU)}$, the evolution operator is $\mathsf{G}(t)=\mathsf{C}(t)\model=\encoder^*_{N, \hypothesis(\domY)}\langle \cdot |\eval_{u_t}\rangle \encoder^*_{N, \hypothesis(\domU)} (\encoder_N\encoder_N^* + \gamma \identity_N)^{-1} \target_N$,

    \noindent where $\langle \cdot |\eval_{u_t}\rangle \encoder^*_{N, \hypothesis(\domU)}=[\langle \eval_{u_i}, \eval_{u_t}\rangle]_{i=1}^N$
\end{example}
The estimator from Example~\ref{example:CKOEstimator} recovers the one of~\cite{bevandaNonparametricControlKoopman2025}. As a special case, let $\domU=\{u_1, \dots,u_m\}$ be finite, and choose the Kronecker kernel $\delta(u_i,u_j)=\{1: u_i{=}u_j, 0:u_i{\neq}u_j$ to represent the inner product $\langle\eval_{u_i}|\eval_{u_j}\rangle$. The estimator~\eqref{eq:EstimatorCKO:GramianForm} then decouples into $m$ estimators, one for each $u_i$. This gives a well-defined estimator for the approach in~\cite{nuskeFiniteDataErrorBounds2023}. In the same vein, as in their approach, one can then interpolate the operators. However, this interpolation can be incorporated directly into learning to leverage correlations between inputs by choosing a non-diagonal kernel $k(u_i,u_j)$. Our analysis provides a clear path to extending their bounds to account for small deviations in the data by including regularization, demonstrating that the framework we propose can seamlessly model and improve upon existing approaches.
\begin{example}[TV Koopman Operator]
    Let $I=[0, T]\subset\Natural$ be a finite time interval and let $\hypothesis = \hypothesis(\domY)\otimes\hypothesis(I)$ and $\prehypothesis = \hypothesis(\domY)$ be separable RKHSs.
    The model is the TV Koopman operator $\model = \mathsf{K}: \prehypothesis \to \hypothesis$ from Lemma~\ref{lemma:EFamilyTVSystem},
    the encoder is the evaluation functional $\encoder = \langle \cdot | \eval_{y,t} \rangle: \hypothesis \to \outputspace$,
    and the target maps observables to their successor values $\target = \mathsf{K}\encodertarget: \prehypothesis \to \outputspace$.
    Given data $\data = \{t_i, y_i, y'_i\}_{i=1}^N$,
    the sample encoder is $\encoder_N: \hypothesis \to \domY^N$, $[\encoder_N h]_i = h(y_i, t_i)$,
    and the sample target is $\target_N \in \domY^N$, $[\target_N]_i = y'_i$.
    Learning $\model$ is an instance of Problem~\ref{problem:OperatorInverse}.
    The Tikhonov-regularized best-approximate solution is again given by~\eqref{eq:EstimatorCKO} and~\eqref{eq:EstimatorCKO:GramianForm}, with the encoders substituted.
    The restriction to a finite-time interval $I$ is necessary from a statistical perspective as $\domT$ is not compact. The time evaluation operator, analogous to the control operator $\mathsf{C}(t)$ from Definition~\ref{definition:ControlOperator}, does not need to be learned.
    The evolution operator becomes $\mathsf{G}(t)=(\identity\otimes\langle \cdot |\eval_{t}\rangle)\model=\encoder^*_{N, \hypothesis(\domY)}\langle \cdot |\eval_{t}\rangle \encoder_{N, \hypothesis(I)} (\encoder_N\encoder_N^* + \gamma \identity_N)^{-1} \target_N$.
\end{example}
Both estimators become computational when the domains are specified. A possible choice is again separable spaces $\hypothesis=\hypothesis (\domY)\otimes \hypothesis (\domT)$ with $\hypothesis (\domT)$ modeling, i.e., continuous or periodic functions. The estimator on finite-dimensional spaces $\prehypothesis, \hypothesis$ is~\eqref{eq:EstimatorCKO}, otherwise~\eqref{eq:EstimatorCKO:GramianForm}, see~\cite{bevandaNonparametricControlKoopman2025} for further details. This emphasizes again that the analysis encompasses both RKHSs defined via finite-dimensional embeddings and infinite-dimensional universal RKHSs
Although we provide examples of Tikhonov regularized estimators, one can choose from a variety of regularization schemes~\cite{englRegularizationInverseProblems1996}.

\subsection{Classification of Approaches from the Literature}
\begin{table*}[t]
    \caption{Classification of selected approaches by their analysis of the bias-Variance trade-offs in the error decomposition~\eqref{eq:ErrorDecomposition}}
    \label{tab:tradeoffs}
    \centering
    \footnotesize \begin{tabular}{@{}lccc@{}}
        \toprule
        \textbf{Method}                                                       & \textbf{Representation}
                                                                              & \textbf{Approximation}
                                                                              & \textbf{Regularization}                                                         \\
                                                                              & \footnotesize{$\|\target{-}\world\|$
                                                                                    \textit{vs.}
                                                                                    $\|\encoder\model^\dagger{-}\target\|$}
                                                                              & \footnotesize{$\|\model^\dagger{-}\model_m^\dagger\|$
                                                                                    \textit{vs.}
                                                                                    remaining}
                                                                              & \footnotesize{reg.\ bias
                                                                                    \textit{vs.}
                                                                                    $\|\regularization_{\gamma,m}\|\cdot\delta$}                                \\
        \midrule
        \multicolumn{4}{@{}l}{\textit{Autonomous}}                                                                                                              \\
        \addlinespace[2pt]
        Williams et al.~\cite{williamsDataDrivenApproximationKoopman2015}     & $\times$                                              & \checkmark & $\times$   \\
        Kostic et al.\cite{kosticSharpSpectralRates2023}                      & \checkmark                                            & \checkmark   & \checkmark \\
        Mollenhauer~\cite{mollenhauerStatisticalApproximationConditional2022} & \checkmark                                            & \checkmark & \checkmark \\
        \midrule
        \multicolumn{4}{@{}l}{\textit{Control}}                                                                                                                 \\
        \addlinespace[2pt]
        Gr{\"u}nwälder et al.~\cite{grunewalderModellingTransitionDynamics2012}   & \checkmark                                            & \checkmark & \checkmark \\
        N{\"u}ske et al.~\cite{nuskeFiniteDataErrorBounds2023}                     & $\times$                                              & \checkmark & $\times$   \\
        Haseli and Cortes~\cite{haseliModelingNonlinearControl2023}           & $(\checkmark)$                                              & \checkmark & $\times$   \\
        Bevanda et al.~\cite{bevandaNonparametricControlKoopman2025}           & \checkmark                                            & \checkmark & \checkmark \\
        \bottomrule
    \end{tabular}%
\end{table*}

The framework we described can be used to classify learning approaches from the literature in terms of them addressing the errors from Proposition~\ref{prop:ErrorDecomposition}.

When it comes to modeling autonomous systems, advances and effective algorithms have been constructed by reconciling operator theory with learning theory and inverse problems, see~\cite{kosticLearningDynamicalSystems2022, bevandaKoopmanKernelRegression2023, devergneBiasedUnbiasedDynamics2024, kosticConsistentLongTermForecasting2024, giannakisPhysicsinformedSpectralApproximation2025a} for recent examples on algorithms with convergence guarantees and strategies explicitly built for data sampling measures.

Modeling systems with exogenous inputs via the Koopman operator approach~\cite{koopmanHamiltonianSystemsTransformation1931}, has received increased attention in the control community throughout the last decades~\cite{mezicApplicationsSpectralTheory2015, bruntonDataDrivenScienceEngineering2019, bevandaKoopmanOperatorDynamical2021, strasserOverviewKoopmanbasedControl2026}. Within this paradigm, the goal is to build a model according to Ansatz~\ref{ansatz:EvolutionOfObservable} in the special case of the operator $\mathsf{K}$ being defined via Ansatz~\ref{ansatz:EvolutionOfOutput}.
The operator $[\mathsf{K}h](x, u)=h(f(x, u))=h\circ f (x, u)$ then becomes a composition operator as in Definition~\ref{definition:ControlKoopmanOperator}.
This approach is naturally encompassed in the framework we introduced through the control evolution operator, cf. Proposition~\ref{prop:EvolutionOperatorControlSystem}.
In terms of deriving structures for learning in control systems, the current literature often fails to reconcile operator theory and learning theory or inverse problems, leading to algorithms that are not well defined in the data limit\footnote{For RKHS methods, Gram matrices grow with sample size. As the population limit is a compact kernel integral operator. Compactness implies that the spectrum decays to zero. Thus, inverting Gram matrices becomes increasingly unstable, as in the limit, no bounded inverse exists.}
or neglecting possible representation error, cf. Remark~\ref{remark:RepresentationError}. See Table~\ref{tab:tradeoffs} for a selection of recent approaches that account for the bias-variance trade-offs in the error decomposition~\eqref{eq:ErrorDecomposition} when deriving guarantees. We note that, while~\cite{haseliModelingNonlinearControl2023} optimizes for representation error, they do so within a finite-dimensional subspace. It would thus be interesting whether their methods can be augmented to allow for convergent estimation of $\world$.

\section{Conclusions}
In this paper, we proposed a structured approach to learning linear operators for dynamical systems.
The design process consists of three steps: Derive the structure of the evolution operator using the theory of evolution equations, choose a structure amenable to learning, and derive and analyze learning algorithms as solutions to the associated inverse problem.
We demonstrated that autonomous, time-varying, and controlled dynamics share the common structure of evolution families and semigroups. We identified structures for learning and analysis, in particular, the control Koopman operator and the evolution semigroup on sequence spaces.
We demonstrated the power of composing both frameworks by deriving novel estimators for time-varying systems as direct instances of the general construction, enabling the use of standard algorithms.
We analyzed learning algorithms through an error decomposition and identified which bias-variance trade-offs existing methods account for and which they neglect.
In particular, we identify algorithms for control systems that are inspired by EDMD, which focus on approximation but do not account for representation and regularization trade-offs. This provides a promising avenue to make future iterations of these algorithms more robust.

\emph{Limitations.}
We identified important trade-offs, but did not quantify them.
Further, while the estimators are guaranteed to be Hilbert-Schmidt, they need not inherit the structural properties of the true operator, such as positivity, unitarity, dissipativity, or the Markov property.
Downstream algorithms may thus fail despite a small error.

\emph{Future Work.}
As stated, extending to infinitesimal generators and their semigroups is natural in the proposed framework. The evolution semigroup structure provides an avenue for stability analysis of time-varying, forced or controlled models. Following up, studying how stability can be certified for the true system from data is of interest. Finally, it seems crucial to leverage the trade-off between representation and estimation error in order to obtain good representations.
\section{ACKNOWLEDGMENTS}
This work is supported by the DAAD programme Konrad
Zuse Schools of Excellence in Artificial Intelligence, sponsored by the Federal Ministry of Education and Research, by the European Union’s Horizon Europe innovation action program under grant agreement No.101093822, "SeaClear2.0" and funded by the Deutsche Forschungsgemeinschaft (DFG, German Research Foundation) – Project 535860958: Active Learning for Systems and Control (ALeSCo) -- Data Informativity, Uncertainty, and Guarantees.

\bibliographystyle{abbrv}
\bibliography{literature}

\end{document}